\documentclass[final,hidelinks]{siamart251216}

\usepackage{graphicx}
\usepackage{latexsym,color}
\usepackage{amssymb}
\usepackage{amsfonts}
\usepackage{amsmath}
\usepackage{microtype}

\newcommand{\E}{\mathbb{E}}
\newcommand{\R}{\mathbb{R}}

\newcommand{\Cov}{\operatorname{Cov}}
\newcommand{\Var}{\operatorname{Var}}

\numberwithin{equation}{section}
\hypersetup{
  pdftitle={High-Frequency Exponential-Utility Maximization under Fractional Brownian Motion},
  pdfauthor={Yan Dolinsky},
  pdfkeywords={Fractional Brownian motion, high frequency, utility maximization}
}

\begin{document}

\title{High-Frequency Exponential-Utility Maximization under Fractional Brownian Motion}
\author{Yan Dolinsky\footnotemark[1] }
\date{\today}
\maketitle
\renewcommand{\thefootnote}{\fnsymbol{footnote}}
\footnotetext[1]
{Department of Statistics, Hebrew University of Jerusalem.
             \email{yan.dolinsky@mail.huji.ac.il}. supported by the ISF Grant 305/25.}
\renewcommand{\thefootnote}{\arabic{footnote}}

\pagenumbering{arabic}

\begin{abstract}
We study utility maximization for high-frequency trading in fractional Brownian
motion models. We first consider exponential utility in the discretized
arithmetic model. Using spectral methods for stationary Gaussian sequences,
we derive the asymptotic growth rate of the optimal certainty equivalent and
the scaling limit of the optimal positions. We then consider power and
logarithmic utility in a positive-price model with positive wealth and
investment fractions in $[0,1]$. In this case, the logarithm of the optimal
certainty equivalent grows at rate $n^{1-H}$ for $H\neq1/2$, with a
coefficient determined by the one-sided prediction error of fractional
Gaussian noise. We also give a simple binary strategy which attains this
leading rate, without claiming exact finite-horizon optimality.
\end{abstract}

\section{Introduction and Main Results}
Fractional Brownian motion 
with Hurst parameter $H\in (0,1)$ is
a continuous, zero-mean Gaussian process $B^H=(B^H_t)_{t\geq 0}$ such that
\begin{equation*}
\Cov\left(B^H_t,B^H_u\right)=\frac{t^{2H}+u^{2H}-|t-u|^{2H}}{2}, \ \ t,u\geq 0.
\end{equation*}
The process $B^H$ is self-similar, satisfies $B^H_{at}\overset{d}{=}a^H B^H_t$, and
has stationary increments. Moreover, the successive 
increments of $B^H$ are positively correlated for $H>1/2$, negatively correlated for $H<1/2$, while $H = 1/2$ 
recovers the usual Brownian motion with independent increments.

Fractional Brownian motion, which can display the long-range dependence observed in empirical data when $H>1/2$
(see \cite{CKW,M:1968,WTT} and the references therein),
is not a semimartingale when 
$H\neq \frac{1}{2}$ and so, in the frictionless case it leads to arbitrage opportunities (see, for instance, 
\cite{R,C}). If trading is restricted to a discrete trading grid, arbitrage opportunities disappear.

This paper was inspired by \cite{GuasoniMishuraRasonyi} where the authors study high-frequency trading in a market in which the asset price is driven by fractional Brownian motion. They derive 
explicit locally mean--variance optimal trading strategies for a fixed trading frequency and analyze their behavior as the trading frequency tends to infinity. In this work, we provide an analogous 
analysis for exponential utility maximization.

We consider a discretized version of the fractional Brownian motion framework. Namely, we fix a time horizon \(T>0\), 
\(n\in\mathbb N\), and consider a financial market that is active at times
$
\left\{0,\frac{T}{n},\frac{2T}{n},\ldots,T\right\}.
$
The financial market consists of a savings account, which for simplicity bears no interest, and a risky asset \(S=B^H\).
The investor's flow of information is given by the filtration
$
\mathcal F_{\frac{kT}{n}}
:=
\sigma\left\{S_0,S_{\frac{T}{n}},\ldots,S_{\frac{kT}{n}}\right\}$,\
$k=0,1,\ldots,n.$

A trading strategy is a random vector $\gamma=\left(\gamma_{\frac{iT}{n}}\right)_{i=0}^{n-1}$ that is adapted to the above filtration. 
Denote by \(\mathcal A_n\) the set of all trading strategies.
For \(\gamma\in\mathcal A_n\), the corresponding portfolio value at the maturity date is given by
\[
V^\gamma_T
=
\sum_{i=1}^{n}
\gamma_{\frac{(i-1)T}{n}}
\left(
S_{\frac{iT}{n}}-S_{\frac{(i-1)T}{n}}
\right).
\]

The investor's preferences are described by the exponential utility function
\[
u(x)=-\exp(-\alpha x),
\qquad x\in\mathbb R,
\]
with absolute risk-aversion parameter \(\alpha>0\), and her goal is to
\begin{equation}\label{2.2}
\operatorname{maximize}\quad
\E
\left[
-\exp\left(-\alpha V^\gamma_T\right)
\right]
\quad\text{over}\quad
\gamma\in\mathcal A_n.
\end{equation}
Clearly, for any trading strategy \(\gamma\) and any constant \(\lambda\in\mathbb R\),
$
V^{\lambda\gamma}_T=\lambda V^\gamma_T.
$
Thus, without loss of generality, we set the risk-aversion parameter to
\(\alpha=1\). Moreover, by the self-similarity of fractional Brownian
motion, we may also assume, without loss of generality, that \(T=1\).

Next, before we formulate our main results we need some preparations.

\subsection{Spectral Density and Asymptotic Formulas}

For the spectral analysis, it is convenient to work with a two-sided fractional
Brownian motion. We therefore extend $B^H$, originally defined on
$\mathbb R_+$, to a centered Gaussian process
$B^H=(B_t^H)_{t\in\mathbb R}$ on the whole real line, with covariance function
\[
\E[B_t^H B_s^H]
=
\frac12\left(|t|^{2H}+|s|^{2H}-|t-s|^{2H}\right),
\qquad s,t\in\mathbb R.
\]
Its restriction to $\mathbb R_+$ has the same law as the fractional Brownian
motion introduced above.

Fractional Gaussian noise is the discrete-time increment process associated
with this two-sided fractional Brownian motion. More precisely, define
\begin{equation}\label{4}
X_j:=B_j^H-B_{j-1}^H,
\qquad j\in\mathbb Z.
\end{equation}
Then $(X_j)_{j\in\mathbb Z}$ is a centered stationary Gaussian sequence with
covariance function
\[
\rho_H(k)
:=
\E[X_0X_k]
=
\frac12\left(
|k+1|^{2H}-2|k|^{2H}+|k-1|^{2H}
\right),
\qquad k\in\mathbb Z.
\]

Since $(X_j)_{j\in\mathbb Z}$ is stationary, its covariance function admits a
spectral representation; see, for instance,
\cite[Chapter~8]{BrockwellDavis}. Its spectral density is the $2\pi$-periodic
function
\[
f_H(\lambda)
:=
\sum_{k\in\mathbb Z}\rho_H(k)e^{-ik\lambda},
\qquad \lambda\in(-\pi,\pi)
\]
and has the explicit representation; see \cite{Sinai},
\begin{equation*}
f_H(\lambda)
=
2\Gamma(2H+1)\sin(\pi H)(1-\cos\lambda)
\sum_{m\in\mathbb Z}|\lambda+2\pi m|^{-1-2H},
\qquad 0<|\lambda|<\pi.
\end{equation*}

The spectral density determines the asymptotic behavior of the covariance
matrices
\begin{equation}\label{2}
\Gamma_n^H
:=
\bigl(\rho_H(j-k)\bigr)_{j,k=1}^n.
\end{equation}
Indeed, $f_H>0$ almost everywhere, $f_H\in L^1(-\pi,\pi)$, and
$\log f_H\in L^1(-\pi,\pi)$, because
$f_H(\lambda)\asymp|\lambda|^{1-2H}$ near zero and $f_H$ is bounded
above and away from zero outside any neighborhood of zero. The first
Szeg\H{o} limit theorem for nonnegative integrable symbols, in the form
given in \cite[Theorem~1.7]{GuoLiZhou} with $d=1$ and
$\sigma_n=\{1,\ldots,n\}$, therefore yields
\begin{equation}\label{szego-limit}
\lim_{n\to\infty}\frac1n\log\det\Gamma_n^H
=
\frac1{2\pi}\int_{-\pi}^{\pi}\log f_H(\lambda)\,d\lambda.
\end{equation}
Thus, the exponential growth rate of the determinant is determined by the
geometric mean of the spectral density.

Moreover, the spectral representation can be used to compute the error of
interpolating $X_0$ from all the remaining variables $(X_j)_{j\neq0}$. More
precisely, the classical interpolation formula for stationary processes gives
(see \cite{Kolmogorov:41,Salehi})
\begin{equation}\label{interpolation-formula}
\Var\bigl(X_0\mid X_j,\ j\in\mathbb Z\setminus\{0\}\bigr)
=
\left(
\frac1{2\pi}\int_{-\pi}^{\pi}\frac{d\lambda}{f_H(\lambda)}
\right)^{-1}.
\end{equation}
The integrals on the right-hand sides of \eqref{szego-limit} and
\eqref{interpolation-formula} are finite because
$f_H(\lambda)\asymp|\lambda|^{1-2H}$ as $\lambda\to0$.

\subsection{Main Results}

We start with the first limit theorem.

\begin{theorem}\label{thm1}
The normalized certainty equivalent under high-frequency trading with
fractional Brownian motion satisfies
\begin{align}\label{certainty-equivalent-limit}
&\lim_{n\to\infty}
-\frac1n
\log\left(
\inf_{\gamma\in\mathcal A_n}
\E\left[e^{-V_1^\gamma}\right]
\right)
\notag\\
&\qquad=
\frac12\left[
\frac1{2\pi}\int_{-\pi}^{\pi}\log f_H(\lambda)\,d\lambda
+
\log\left(
\frac1{2\pi}\int_{-\pi}^{\pi}\frac{d\lambda}{f_H(\lambda)}
\right)
\right].
\end{align}
\end{theorem}

We note that the right-hand side of
\eqref{certainty-equivalent-limit} can be written as
\[
\frac12\log\left(
\frac{\operatorname{AM}(g_H)}
{\operatorname{GM}(g_H)}
\right),
\qquad g_H:=\frac{1}{f_H},
\]
where
\[
\operatorname{AM}(g_H)
:=
\frac{1}{2\pi}\int_{-\pi}^{\pi}g_H(\lambda)\,d\lambda
\]
is the average value of \(g_H\), and
\[
\operatorname{GM}(g_H)
:=
\exp\left(
\frac{1}{2\pi}\int_{-\pi}^{\pi}\log g_H(\lambda)\,d\lambda
\right)
\]
is the corresponding geometric mean.

The proof is based on the main result of \cite{DZ:2023}, which characterizes,
in the Gaussian framework, the value of the exponential utility maximization problem
\eqref{2.2} and the corresponding unique optimal portfolio. 
We arrive at the following limit theorem 
which is closely
related to Theorem 2.3 in \cite{GuasoniMishuraRasonyi}.

\begin{theorem}\label{thm2}
For each $n\in\mathbb{N}$, let $\gamma^n\in\mathcal{A}_n$ be the unique
optimal portfolio for the utility maximization problem \eqref{2.2} in the
$n$-step market. Then, as $n\to\infty$, the sequence of Gaussian processes
\[
\left(n^{-H}\gamma^n_{\frac{\lfloor nt\rfloor}{n}}\right)_{t\in(0,1)}
\]
converges in finite-dimensional distributions to a centered Gaussian
process $(B_t)_{t\in(0,1)}$,
such that $B_s$ and $B_t$ are independent for every
$0<s<t<1$, and
\begin{equation*}
\operatorname{Var}(B_t)=
\frac{1}{2\pi}
\int_{-\pi}^{\pi}
f_H(\lambda)
\left|
\sum_{k=1}^{\infty}
\left(
\frac{1}{2\pi}
\int_{-\pi}^{\pi}
\frac{e^{iku}}{f_H(u)}\,du
\right)e^{-ik\lambda}
\right|^2
d\lambda,
\qquad t\in(0,1).
\end{equation*}
\end{theorem}

The limiting process may be interpreted as white-noise-type Gaussian field: its values at
distinct positive times are independent, centered Gaussian random variables
with the same variance. Thus, in the high-frequency limit, the dependence
between optimal positions at different trading times disappears.

\begin{figure}[!ht]
\centering
\includegraphics[width=0.7\columnwidth]{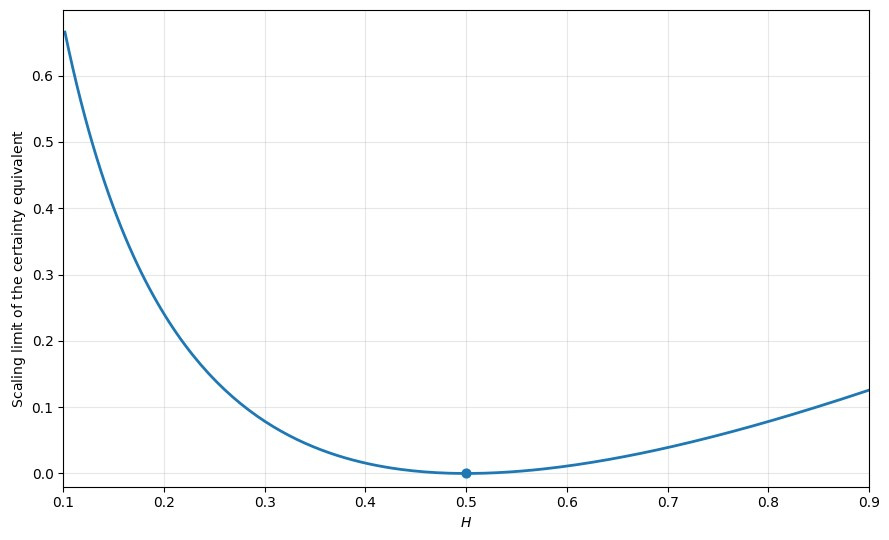}
\caption{\footnotesize
The scaling limit of the normalized certainty equivalent in
\eqref{certainty-equivalent-limit} as a function of the Hurst parameter $H$.
}
\label{fig:value_vs_H}
\end{figure}

\subsection{Power Utility Extension}\label{subsec:power-extension}

The previous results concern exponential utility in the arithmetic model
$S=B^H$. We now add a second problem in order to clarify which high-frequency
conclusions are specific to that setting. A direct power-utility formulation
in the arithmetic model is not meaningful under the usual requirement of
positive wealth: any nonzero final holding produces conditionally Gaussian
terminal wealth with unbounded downside, so almost-sure positivity forces
the last holding to vanish, and backward induction eliminates every holding.
We therefore denote the positive price process by $\widetilde S$:
\begin{equation}\label{positive-price}
\widetilde S_t=\exp\left(B_t^H-\frac12 t^{2H}\right),
\qquad 0\leq t\leq1,
\end{equation}
while keeping the same discrete observation filtration and a zero-interest
savings account. Indeed, observing $\widetilde S$ on the grid is equivalent
to observing $B^H$ there.

Let $x>0$ be the initial wealth. A trading strategy
$\pi=(\pi_{i-1})_{i=1}^n$ is predictable with respect to the discrete observation
filtration. We denote by $W_i^\pi$ the total value of the self-financing
portfolio at time $i/n$, including both the risky asset and the savings
account, and set $W_0^\pi=x$. At time $(i-1)/n$, the investor allocates the
fraction $\pi_{i-1}$ of her current wealth $W_{i-1}^\pi$ to the risky asset
and keeps the remaining fraction $1-\pi_{i-1}$ in the savings account.
Thus, she holds $\pi_{i-1}W_{i-1}^\pi/\widetilde S_{(i-1)/n}$ shares over the next
trading period. Since the savings account bears no interest, the portfolio
value at the end of that period is
\begin{equation*}
\begin{aligned}
W_i^\pi
&=(1-\pi_{i-1})W_{i-1}^\pi
  +\frac{\pi_{i-1}W_{i-1}^\pi}{\widetilde S_{(i-1)/n}}\widetilde S_{i/n}\\
&=W_{i-1}^\pi\left(1-\pi_{i-1}+\pi_{i-1}e^{Y_i}\right),
\qquad i=1,\ldots,n,
\end{aligned}
\end{equation*}
where $Y_i:=\log(\widetilde S_{i/n}/\widetilde S_{(i-1)/n})$ is the risky asset's log return.
In particular, $W_n^\pi$ is the terminal portfolio value at time $1$.

The conditional gross return has support $(0,\infty)$, so strict positivity
of the portfolio value at every trading date (for finite predictable fractions) is equivalent to
$0\leq\pi_{i-1}\leq1$ almost surely. We call such strategies admissible.
These are investment fractions, rather than the numbers of shares used in
the original model. Thus, the comparison changes the price specification
and admissible portfolios as well as utility.

For $U_p(w)=w^p/p$, $p<1$, $p\neq0$, relative risk aversion is $1-p$.
For $p=0$ we use logarithmic utility $U_0(w)=\log w$. Define the log return
of terminal wealth by $L_n^\pi:=\log(W_n^\pi/x)$ and the log
certainty-equivalent gain of a strategy by
\[
J_{n,p}(\pi):=
\begin{cases}
\displaystyle\frac1p\log\E\left[e^{pL_n^\pi}\right],&p\neq0,\\[4pt]
\E[L_n^\pi],&p=0,
\end{cases}
\qquad
G_{n,p}:=\sup_\pi J_{n,p}(\pi),
\]
where the supremum is over all admissible strategies. 
Define
\[
\nu_H:=\exp\left(\frac1{2\pi}\int_{-\pi}^{\pi}
\log f_H(\lambda)\,d\lambda\right),
\qquad
a_H:=\sqrt{\frac{1-\nu_H}{2\pi}}.
\]
As shown in Section~\ref{sec:power-proof}, $\nu_H$ is the variance of the
error in predicting fractional Gaussian noise from its entire past, and
$0<\nu_H\leq1$.

\begin{theorem}\label{thm3}
For every fixed $H\in(0,1)$ and $p<1$, including the logarithmic case,
\begin{equation}\label{power-limit}
\frac{G_{n,p}}{n^{1-H}}\longrightarrow a_H.
\end{equation}
For each $n$, define the observed, standardized grid increments and their
one-step forecasts by
\[
X_i^{(n)}:=n^H\left(B_{i/n}^H-B_{(i-1)/n}^H\right),
\qquad
m_i^{(n)}:=\E\left[X_i^{(n)}\mid\mathcal F_{(i-1)/n}\right],
\quad 1\leq i\leq n.
\]
Then the predictable binary strategy
$\widehat\pi_{i-1}^{(n)}:=\mathbf{1}_{\{m_i^{(n)}>0\}}$ satisfies
\[
\frac{J_{n,p}(\widehat\pi^{(n)})}{n^{1-H}}\longrightarrow a_H.
\]
In particular, this strategy is asymptotically optimal at the leading
$n^{1-H}$ scale.
\end{theorem}

The constant in \eqref{power-limit} involves one-sided prediction, in
contrast to the two-sided interpolation entering
\eqref{certainty-equivalent-limit}. It is positive for $H\neq1/2$.
At $H=1/2$, the price process in \eqref{positive-price} is a martingale,
and concavity of the utility function shows that cash is exactly optimal,
so $G_{n,p}=0$ for every $n$. The binary rule is only claimed to be
leading-rate optimal: exact finite-horizon optimizers may use interior
fractions and, for nonlogarithmic power utility, intertemporal hedging
positions. The contrast between the two growth rates is a comparison of
different wealth and trading constraints, not an effect of preferences
alone; see Section~\ref{sec:discussion}.

\section{Proof of the Main Results}
We first prove Theorem~\ref{thm1}.

\begin{proof}[Proof of Theorem~\ref{thm1}]
We apply Corollary~1.3(I) of \cite{DZ:2023}. By the self-similarity of fractional Brownian motion, we obtain
\begin{equation}\label{finite-n-value}
\inf_{\gamma\in\mathcal{A}_n}
\E\left[\exp\left(-V_1^\gamma\right)\right]
=
\left(
\det\left(\Gamma_n^H\right)
\prod_{i=1}^n
\left(\Lambda_n^H\right)_{ii}
\right)^{-1/2},
\end{equation}
where $\Gamma_n^H$ is given by~\eqref{2} and $\Lambda_n^H:=(\Gamma_n^H)^{-1}$ is the corresponding precision matrix. 

In view of \eqref{szego-limit} and \eqref{finite-n-value}, it remains to show that 
\begin{equation}\label{precision-limit-proof}
\lim_{n\to\infty}
\frac1n\sum_{i=1}^n\log(\Lambda_n^H)_{ii}
=
\log\left(
\frac1{2\pi}\int_{-\pi}^{\pi}\frac{d\lambda}{f_H(\lambda)}
\right).
\end{equation}
To this end,
for $p,q\in\mathbb Z_+$, define the finite interpolation error
\[
v_{p,q}
:=
\Var\bigl(
X_0\mid X_{-p},\ldots,X_{-1},X_1,\ldots,X_q
\bigr),
\]
where $X_j$, $j\in\mathbb Z$, are given by \eqref{4}; an empty conditioning block is omitted.

From the stationarity of 
the Gaussian sequence $(X_j)_{j\in\mathbb Z}$, it follows (see Section 2 in \cite{RH:05}) that 
\begin{equation}\label{precision-as-interpolation}
(\Lambda_n^H)_{ii}
=
\frac1{v_{i-1,n-i}},
\qquad i=1,\ldots,n.
\end{equation}

Let
\[
\mathcal G_m
:=
\sigma(X_{-m},\ldots,X_{-1},X_1,\ldots,X_m),
\qquad
\mathcal G_\infty
:=
\sigma(X_j:\ j\in\mathbb Z\setminus\{0\}).
\]
Then $\mathcal G_m\uparrow\mathcal G_\infty$. By the $L^2$ martingale
convergence theorem,
$\E[X_0\mid\mathcal G_m]$ converges in $L^2$ to
$\E[X_0\mid\mathcal G_\infty]$, and hence
\begin{equation}\label{finite-errors-converge}
v_{m,m}\downarrow v_\infty
:=
\Var(X_0\mid\mathcal G_\infty).
\end{equation}
Formula \eqref{interpolation-formula} identifies this limiting error as
\[
v_\infty
=
\left(
\frac1{2\pi}\int_{-\pi}^{\pi}\frac{d\lambda}{f_H(\lambda)}
\right)^{-1}>0.
\]
Also, conditioning can only decrease variance, so
\begin{equation}\label{variance-bounds-proof}
v_\infty\leq v_{p,q}\leq\Var(X_0)=1,
\qquad p,q\in\mathbb Z_+.
\end{equation}

Fix $m\in\mathbb N$. If $m+1\leq i\leq n-m$, then the conditioning family
appearing in $v_{i-1,n-i}$ contains the one appearing in $v_{m,m}$. Hence,
by \eqref{variance-bounds-proof},
\[
v_\infty\leq v_{i-1,n-i}\leq v_{m,m}.
\]
Using \eqref{precision-as-interpolation}, we obtain
\[
-\log v_{m,m}
\leq
\log(\Lambda_n^H)_{ii}
\leq
-\log v_\infty
\]
for every $m+1\leq i\leq n-m$. For all $i=1,\ldots,n$,
\eqref{variance-bounds-proof} also gives
\[
0\leq\log(\Lambda_n^H)_{ii}\leq-\log v_\infty.
\]
There are $n-2m$ interior indices, and therefore, for $n>2m$,
\[
\frac{n-2m}{n}(-\log v_{m,m})
\leq
\frac1n\sum_{i=1}^n\log(\Lambda_n^H)_{ii}
\leq
-\log v_\infty.
\]
Finally, by taking $n\rightarrow\infty$ and then $m\rightarrow\infty$,
\eqref{finite-errors-converge} gives \eqref{precision-limit-proof}
and the proof is complete.
\end{proof}

Next, we prove Theorem \ref{thm2}.
\begin{proof}[Proof of Theorem~\ref{thm2}]
Fix $n\in\mathbb N$. Corollary~1.3(I) in \cite{DZ:2023} gives the optimal portfolio $\gamma^n\in\mathcal A_n$ as 
\begin{equation}\label{optimal-strategy}
n^{-H}\gamma^n_{\frac{i-1}{n}}
=
-\sum_{j=1}^{i-1}(\Lambda_n^H)_{ij}X_j,
\qquad i=1,\ldots,n.
\end{equation}

By self-similarity, the normalized grid increments in the $n$-step market
have the law of $(X_1,\ldots,X_n)$. In the rest of this proof we work on this
law-equivalent realization, so that the following $L^2$ comparisons are
between random variables on the same probability space.

For $1\leq i\leq n$, define
\[
Q_{n,i}(\lambda)
:=
\sum_{r=1-i}^{n-i}
(\Lambda_n^H)_{i,i+r}e^{ir\lambda}.
\]
The identity $\Lambda_n^H\Gamma_n^H=I_n$ gives 
\[
\frac{1}{2\pi}
\int_{-\pi}^{\pi}
Q_{n,i}(\lambda)e^{-is\lambda}
f_H(\lambda)\,d\lambda
=
\delta_{0s}, \ \ s=1-i,\ldots,n-i. 
\]
Here, $\delta_{0s}$ denotes the Kronecker delta, defined by
\[
\delta_{0s}
=
\begin{cases}
1, & s=0,\\
0, & s\neq 0.
\end{cases}
\]
The obvious relation 
\[
\frac{1}{2\pi}
\int_{-\pi}^{\pi}
\frac{1}{f_H(\lambda)}e^{-is\lambda}
f_H(\lambda)\,d\lambda
=
\delta_{0s}, \ \ s=1-i,\ldots,n-i
\]
yields that $Q_{n,i}$ is the $\|\cdot\|_{f_H}$-orthogonal projection of $1/f_H$
onto the span of the frequencies $1-i,\ldots,n-i$.
Thus,
\begin{equation}\label{eq:Q-conv-short}
\min(i_n,n-i_n)\to\infty \ \Rightarrow \ Q_{n,i_n}\longrightarrow\frac{1}{f_H}
\ \ \text{in} \ \ \|\cdot\|_{f_H}.
\end{equation}

Let \(P_-\) denote the Fourier projection onto the strictly negative
frequencies, i.e.,
\[
P_-g(\lambda)
:=
\sum_{k<0}\widehat g(k)e^{ik\lambda},
\qquad
\widehat g(k):=\frac{1}{2\pi}\int_{-\pi}^{\pi}
g(\lambda)e^{-ik\lambda}\,d\lambda.
\]

The relation $f_H(\lambda)\asymp|\lambda|^{1-2H}$ as $\lambda\to0$,
together with the fact that $f_H$ is bounded above and away from zero outside
a neighborhood of zero, implies that $f_H$ is an $A_2$ weight because
$1-2H\in(-1,1)$. Hence, by Theorem~1 in
\cite{HuntMuckenhouptWheeden}, $P_-$ is bounded in the norm
$\|\cdot\|_{f_H}$.
We emphasize that its boundedness in the norm
\[
\|g\|_{f_H}^2
:=
\frac{1}{2\pi}
\int_{-\pi}^{\pi}|g(\lambda)|^2f_H(\lambda)\,d\lambda
\]
is not automatic from general Hilbert-space theory. Indeed, an
orthogonal projection onto a closed subspace of a Hilbert space is
always bounded, with operator norm one. However, \(P_-\) is the
orthogonal projection onto the negative Fourier frequencies with
respect to the unweighted \(L^2\)-inner product, and it is generally
not orthogonal with respect to the weighted inner product induced by $f_H$.

Next, define
\[
q_H(\lambda)
:=
P_-\left(\frac{1}{f_H}\right)(\lambda)
=
\sum_{k=1}^{\infty}c_k^He^{-ik\lambda},
\qquad
c_k^H
:=
\frac{1}{2\pi}
\int_{-\pi}^{\pi}
\frac{e^{iku}}{f_H(u)}\,du,
\]
where the Fourier series converges in $L^2(f_H)$. By the boundedness of
$P_-$ and \eqref{eq:Q-conv-short},
\begin{equation}\label{eq:negative-projection-conv}
\left\|P_-Q_{n,i_n}-q_H\right\|_{f_H}
\leq C_H\left\|Q_{n,i_n}-\frac{1}{f_H}\right\|_{f_H}
\longrightarrow0
\end{equation}
whenever $\min(i_n,n-i_n)\to\infty$.

Introduce the centered stationary Gaussian sequence
\[
Y_m:=-\sum_{k=1}^{\infty}c_k^HX_{m-k},
\qquad m\in\mathbb Z,
\]
where the series converges in $L^2$.

Observe that the variance of $Y_m$ is independent of $m$ and equals
\begin{equation}\label{eq:sigma-short}
\sigma_H^2
:=
\frac{1}{2\pi}
\int_{-\pi}^{\pi}
f_H(\lambda)
\left|
\sum_{k=1}^{\infty}
\left(
\frac{1}{2\pi}
\int_{-\pi}^{\pi}
\frac{e^{iku}}{f_H(u)}\,du
\right)e^{-ik\lambda}
\right|^2
d\lambda.
\end{equation}

Moreover,
\[
\operatorname{Cov}(Y_0,Y_m)
=
\frac{1}{2\pi}
\int_{-\pi}^{\pi}
e^{im\lambda}f_H(\lambda)|q_H(\lambda)|^2\,d\lambda.
\]
Since $f_H|q_H|^2$ is integrable, the Riemann--Lebesgue lemma yields
\begin{equation}\label{eq:cov-short}
\operatorname{Cov}(Y_0,Y_m)\longrightarrow0
\qquad\text{as }|m|\to\infty.
\end{equation}

Finally, choose $0<s<t<1$, and let $(i_n)_{n\in\mathbb N}$ and
$(j_n)_{n\in\mathbb N}$ be sequences such that
$i_n/n\to s$ and $j_n/n\to t$. From \eqref{optimal-strategy} and
\eqref{eq:negative-projection-conv}, using the spectral isometry, we obtain
\[
n^{-H}\gamma^n_{\frac{i_n}{n}}-Y_{i_n+1}
\longrightarrow0
\qquad\text{in }L^2,
\]
and likewise with $j_n$ in place of $i_n$. Hence, by
\eqref{eq:sigma-short} and \eqref{eq:cov-short},
\[
\operatorname{Var}\left(
n^{-H}\gamma^n_{\frac{i_n}{n}}
\right)
\longrightarrow\sigma_H^2
\qquad\text{and}\qquad
\operatorname{Cov}\left(
n^{-H}\gamma^n_{\frac{i_n}{n}},
n^{-H}\gamma^n_{\frac{j_n}{n}}
\right)
\longrightarrow0.
\]
This completes the proof. 
\end{proof}

\section{Power Utility: Prediction and Asymptotic Optimality}\label{sec:power-proof}

\begin{proof}[Proof of Theorem~\ref{thm3}]
Fix $n$ and put $\delta=n^{-H}$. Throughout this proof, abbreviate
$X_i^{(n)}$, $m_i^{(n)}$, and $\widehat\pi^{(n)}$ by $X_i$, $m_i$, and
$\widehat\pi$, respectively. By self-similarity, $(X_1,\ldots,X_n)$ has the
same law as the first $n$ variables of the stationary sequence in
\eqref{4}, and hence has covariance matrix $\Gamma_n^H$. Moreover,
$\mathcal F_{i/n}=\sigma(X_1,\ldots,X_i)$ and, by \eqref{positive-price},
\[
Y_i=\delta X_i+d_i,
\qquad
d_i:=-\frac12\left[\left(\frac{i}{n}\right)^{2H}
-\left(\frac{i-1}{n}\right)^{2H}\right].
\]
Thus $d_i\leq0$ and $\sum_{i=1}^n|d_i|=1/2$.

Define the innovations by $\varepsilon_i:=X_i-m_i$ and write
$v_i:=\Var(\varepsilon_i)$. Gaussian orthogonal projection implies that
$\varepsilon_i$ is independent of $\mathcal F_{(i-1)/n}$, the innovations
are mutually independent, and $0<v_i\leq1$. Set
\[
M_n:=\sum_{i=1}^n m_i^+,
\qquad
N_n^\pi:=\sum_{i=1}^n\pi_{i-1}\varepsilon_i,
\qquad
Q_n:=\sum_{i=1}^n X_i^2,
\]
where $z^+:=\max(z,0)$. In particular, $\E[N_n^\pi]=0$ and $\E[Q_n]=n$.

\emph{Wealth bounds.}
For $0\leq\theta\leq1$, the function
$h_\theta(y):=\log(1-\theta+\theta e^y)$ satisfies
$h_\theta(0)=0$, $h_\theta'(0)=\theta$, and
$0\leq h_\theta''(y)\leq1/4$. Consequently,
\[
\theta y\leq h_\theta(y)\leq\theta y+\frac{y^2}{8},
\qquad y\in\R.
\]
Importantly, $h_\theta$ is increasing, so the nonpositive drift can be
discarded \emph{before} applying this bound:
\[
h_{\pi_{i-1}}(Y_i)\leq h_{\pi_{i-1}}(\delta X_i)
\leq\delta\pi_{i-1}X_i+\frac18\delta^2X_i^2.
\]
Since $\pi_{i-1}m_i\leq m_i^+$, summation gives
\begin{equation}\label{log-wealth-upper}
L_n^\pi\leq\delta M_n+\delta N_n^\pi+\frac18\delta^2Q_n.
\end{equation}
For the binary strategy $h_{\widehat\pi_{i-1}}(Y_i)=\widehat\pi_{i-1}Y_i$,
and therefore
\begin{equation}\label{binary-wealth}
L_n^{\widehat\pi}=\delta M_n+\delta N_n^{\widehat\pi}+D_n,
\qquad
D_n:=\sum_{i=1}^n\mathbf{1}_{\{m_i>0\}}d_i\in[-1/2,0].
\end{equation}

\emph{Exponential-moment bounds.}
Conditional Gaussianity and $\pi_{i-1}^2v_i\leq1$ give, by iteration,
\begin{equation}\label{innovation-mgf}
\log\E\left[e^{tN_n^\pi}\right]\leq\frac{t^2n}{2}.
\end{equation}
Set $K_n:=\sup_{u\in[0,1]^n}\Var(\sum_i u_im_i)$. Writing the Gaussian
forecast vector as $AZ$ with $Z$ standard Gaussian, the identity
$M_n=\max_{u\in[0,1]^n}u^\top AZ$ shows that $M_n$ is
$\sqrt{K_n}$-Lipschitz in $Z$. Gaussian concentration
\cite[Chapter~5]{BoucheronLugosiMassart} gives
\begin{equation}\label{forecast-concentration}
\log\E\left[e^{t(M_n-\E[M_n])}\right]\leq\frac{t^2K_n}{2},
\qquad t\in\R.
\end{equation}
Indeed, $m_i=X_i-\varepsilon_i$ and the innovation variance of
$\sum_i u_i\varepsilon_i$ is at most $n$. For $H<1/2$, summability of
$\rho_H$ bounds the variance of $\sum_i u_iX_i$ by $O_H(n)$; for
$H>1/2$, nonnegative covariances bound it by
$\Var(\sum_iX_i)=n^{2H}$. The case $H=1/2$ is immediate. Hence
$K_n=O_H(n+n^{2H})$ and $\delta^2(K_n+n)=O_H(1+n\delta^2)$.

The maximal row sum of $\Gamma_n^H$ is $O_H(1)$ for $H\leq1/2$
and $O_H(n^{2H-1})$ for $H>1/2$, so
$\delta^2\lambda_{\max}(\Gamma_n^H)\to0$. The Gaussian quadratic-form
formula and $\operatorname{tr}\Gamma_n^H=n$ therefore give, for each fixed
$c>0$ and all large $n$,
\begin{equation}\label{quadratic-mgf}
\log\E\left[e^{c\delta^2Q_n}\right]
=-\frac12\log\det(I_n-2c\delta^2\Gamma_n^H)
=O_{H,c}(n\delta^2).
\end{equation}

\emph{Value estimates.}
Let $0<p<1$. Jensen's inequality and \eqref{binary-wealth} imply
\[
G_{n,p}\geq J_{n,p}(\widehat\pi)
\geq\E[L_n^{\widehat\pi}]
\geq\delta\E[M_n]-\frac12.
\]
For the upper bound, apply H\"older with exponents $3,3,3$ to the
three terms in \eqref{log-wealth-upper}, after centering $M_n$. The bounds
\eqref{innovation-mgf}--\eqref{quadratic-mgf} give, uniformly in $\pi$,
\[
\begin{aligned}
J_{n,p}(\pi)
&\leq\delta\E[M_n]+\frac{3p}{2}\delta^2(K_n+n)
 +\frac1{3p}\log\E\left[e^{(3p/8)\delta^2Q_n}\right]\\
&\leq\delta\E[M_n]+O_{H,p}(1+n\delta^2).
\end{aligned}
\]
Taking the supremum yields the upper bound for $G_{n,p}$; the same
bounds apply to $J_{n,p}(\widehat\pi)$.

For logarithmic utility, \eqref{log-wealth-upper} and
\eqref{binary-wealth} directly give
\[
\delta\E[M_n]-\frac12
\leq J_{n,0}(\widehat\pi)
\leq G_{n,0}
\leq\delta\E[M_n]+\frac18n\delta^2.
\]

For negative powers, write $p=-q<0$. Jensen's inequality, with attention
to the negative factor $-1/q$, gives
$J_{n,-q}(\pi)\leq\E[L_n^\pi]$ and hence $G_{n,-q}\leq G_{n,0}$.
For the lower bound, \eqref{binary-wealth}, Cauchy--Schwarz, and
\eqref{innovation-mgf}--\eqref{forecast-concentration} give
\[
\begin{aligned}
\log\E\left[e^{-qL_n^{\widehat\pi}}\right]
&\leq-q\delta\E[M_n]+\frac q2
 +\frac12\log\E\left[e^{-2q\delta(M_n-\E[M_n])}\right]\\
&\qquad+\frac12\log\E\left[e^{-2q\delta N_n^{\widehat\pi}}\right]\\
&\leq-q\delta\E[M_n]+\frac q2+q^2\delta^2(K_n+n).
\end{aligned}
\]
Division by $-q$ reverses the inequality, so
\[
G_{n,-q}\geq J_{n,-q}(\widehat\pi)
\geq\delta\E[M_n]-\frac12-q\delta^2(K_n+n).
\]
Thus, for every fixed $p<1$, including $p=0$,
\begin{equation}\label{power-expansion}
\begin{aligned}
G_{n,p}
&=n^{-H}\sum_{i=1}^n\E[m_i^+]+O_{H,p}(1+n^{1-2H}),\\
J_{n,p}(\widehat\pi)
&=n^{-H}\sum_{i=1}^n\E[m_i^+]+O_{H,p}(1+n^{1-2H}).
\end{aligned}
\end{equation}
Both remainders are $o(n^{1-H})$, because
$(1+n^{1-2H})/n^{1-H}=n^{H-1}+n^{-H}\to0$.

\emph{Identification of the leading constant.}
Orthogonality of $m_i$ and $\varepsilon_i$ gives $\Var(m_i)=1-v_i$,
so Gaussian integration yields
$\E[m_i^+]=\sqrt{(1-v_i)/(2\pi)}$.
By stationarity,
\[
v_i=\Var\left(X_0\mid X_{-(i-1)},\ldots,X_{-1}\right),
\]
where the variables on the right belong to the stationary sequence
\eqref{4}, and the conditioning is empty for $i=1$.
Therefore $v_i$ decreases to the infinite-past prediction error $v_*$,
by $L^2$ convergence of the conditional expectations. Moreover,
$v_i\geq v_\infty>0$, where $v_\infty$ is the two-sided interpolation
error in \eqref{finite-errors-converge}: conditioning on all other
variables supplies at least as much information as conditioning on a
finite past. Thus $\log v_i\to\log v_*$ with a finite limit.

The triangular innovation decomposition has unit diagonal, and hence
$\det\Gamma_n^H=\prod_{i=1}^n v_i$. The first Szeg\H{o} limit theorem
\cite[Theorem~1.7]{GuoLiZhou}, already applied in \eqref{szego-limit}, gives
\[
\frac1n\sum_{i=1}^n\log v_i
=\frac1n\log\det\Gamma_n^H
\longrightarrow\frac1{2\pi}\int_{-\pi}^{\pi}\log f_H(\lambda)\,d\lambda
=\log\nu_H.
\]
By Ces\`aro's convergence theorem, a convergent sequence and its arithmetic
means have the same limit; see \cite[Chapter~III, Section~3.8,
Theorem~12]{Hardy}. Applying it to $(\log v_i)$ shows that
$v_*=\nu_H$, and hence $v_i\downarrow\nu_H$.

Finally, $(2\pi)^{-1}\int_{-\pi}^{\pi}f_H(\lambda)\,d\lambda=\rho_H(0)=1$.
The arithmetic--geometric mean inequality gives $0<\nu_H\leq1$, with
strict inequality when $H\neq1/2$, since then
$\rho_H(1)=2^{2H-1}-1\neq0$ and $f_H$ is not constant.
For $H=1/2$, $f_H\equiv1$ and $\nu_H=1$.
A second application of Ces\`aro's theorem gives
\[
\frac1n\sum_{i=1}^n\E[m_i^+]
\longrightarrow\sqrt{\frac{1-\nu_H}{2\pi}}=a_H.
\]
Combining this limit with \eqref{power-expansion} proves
\eqref{power-limit} and the claimed leading-rate optimality.
\end{proof}

\section{Discussion of the Two Problems}\label{sec:discussion}

For $H\neq1/2$, Theorems~\ref{thm1} and~\ref{thm3} yield growth of order
$n$ for the exponential certainty equivalent, versus $n^{1-H}$ for the
logarithm of the power certainty equivalent relative to initial wealth.
This contrast reflects the different wealth and trading constraints,
rather than preferences alone: the arithmetic model allows unrestricted
positions and negative wealth, whereas the positive-price model imposes
positive wealth and investment fractions in $[0,1]$. The corresponding
spectral constants involve two-sided interpolation and one-sided
prediction, respectively.

The binary rule is optimal only at the leading $n^{1-H}$ scale; it need
not be the exact finite-horizon optimizer. Exact optimal strategies may
involve interior fractions and, for nonlogarithmic power utility, hedging
demands associated with future investment opportunities. Logarithmic
utility is myopic in this model, but its optimal fractions may still be
interior.

\section*{Acknowledgment}

I would like to thank the anonymous referee for a very useful report.
In particular, I thank the referee for formulating the power-utility problem
considered in this paper and for providing a detailed sketch of the solution.

\end{document}